\documentclass{llncs}
\usepackage{amsmath}
\usepackage{amssymb}
\usepackage{enumerate}
\usepackage{varioref}
\usepackage{charter,eulervm}
\usepackage{graphicx}
\usepackage{subfig,wrapfig}
\usepackage{boxedminipage}

\title{{\sc Independent Set in Categorical Products of Cographs 
and Splitgraphs}}
\author{
 Wing-Kai~Hon\inst{1}
\and
 Ton~Kloks\inst{1}
\and 
 Hsiang~Hsuan~Liu\inst{1}
\and 
 Sheung-Hung~Poon\inst{1} 
\and
 Yue-Li~Wang\inst{2}} 
\institute{
 Department of Computer Science\\
 National Tsing Hua University, Taiwan\\
 {\tt \{wkhon,hhliu,kloks,spoon\}@cs.nthu.edu.tw} 
\and 
 Department of Information Management\\ 
 National Taiwan University of Science and Technology\\ 
 {\tt ylwang@cs.ntust.edu.tw}
}

\begin{document}

\maketitle

\begin{abstract}
We show that there are polynomial-time 
algorithms to compute maximum independent sets  
in the categorical products of two cographs and 
two splitgraphs. We show that the ultimate categorical 
independence ratio is computable in polynomial time 
for cographs.  
\end{abstract}

\section{Introduction}
\label{section intro}
%%%%%%%%%%%%%%%%%%%%%%

Let $G$ and $H$ be two graphs. The categorical product also 
goes under the name of tensor product, or direct product, 
or Kronecker product, and even more names have been given to it. 
It is defined as follows. It is a graph, denoted as 
$G \times H$. Its vertices are the ordered pairs 
$(g,h)$ where $g \in V(G)$ and $h \in V(H)$. Two of its 
vertices, say $(g_1,h_1)$ and $(g_2,h_2)$ are adjacent 
if 
\[\{\;g_1,\;g_2\;\} \in E(G) 
\quad\text{and}\quad \{\;h_1,\;h_2\;\} \in E(H).\]  

\bigskip 

One of the reasons for its popularity is Hedetniemi's 
conjecture, which is now more 
than 40 years old~\cite{kn:hedetniemi,kn:tardif,kn:sauer,kn:zhu}. 
\begin{conjecture}
For any two graphs $G$ and $H$ 
\[\chi(G \times H) = \min\;\{\;\chi(G),\;\chi(H)\;\}.\] 
\end{conjecture}
It is easy to see that the right-hand side is an upperbound. 
Namely, if $f$ is a vertex coloring of $G$ then one can color 
$G \times H$ by defining a coloring $f^{\prime}$ 
as follows 
\[f^{\prime}((g,h))=f(g), \quad\text{for all $g \in V(G)$ and $h \in V(H)$.}\]  
Recently, it was shown that the fractional version of 
Hedetniemi's conjecture is true~\cite{kn:zhu2}. 

\bigskip 

When $G$ and $H$ are perfect then Hedetniemi's conjecture 
is true. Namely, let $K$ be a clique of cardinality 
at most 
\[|K| \leq \min \;\{\;\omega(G),\;\omega(H)\;\}.\] 
It is easy to check that $G \times H$ has a clique 
of cardinality $|K|$. One obtains an `elegant' proof, 
via homomorphisms, as follows.  
By assumption, there exist homomorphisms $K \rightarrow G$ and 
$K \rightarrow H$. 
This implies that there also is a 
homomorphism $K \rightarrow G \times H$ (see, eg,~\cite{kn:hahn,kn:hell}). 
(Actually, if $W$, $P$ and $Q$ are any graphs, 
then there exist homomorphisms $W \rightarrow P$ and 
$W \rightarrow Q$ if and only if there exists a 
homomorphism $W \rightarrow P \times Q$.)
In other words~\cite[Observation~5.1]{kn:hahn}, 
\[\omega(G \times H) \geq \min\;\{\;\omega(G),\;\omega(H)\;\}.\] 
Since $G$ and $H$ are perfect, $\omega(G)=\chi(G)$ and 
$\omega(H)=\chi(H)$. This proves the claim, since 
\begin{equation}
\begin{split}
\chi(G \times H) \geq \omega(G \times H) \geq 
& \min\;\{\;\omega(G),\;\omega(H)\;\} \\
= & \min\;\{\;\chi(G),\;\chi(H)\;\} 
\geq \chi(G \times H).
\end{split}
\end{equation}

\bigskip 

Much less is known about the independence 
number of $G \times H$. It is easy to see that 
\begin{equation}
\label{eqn0}
\alpha(G \times H) \geq \max\;\{\;\alpha(G) \cdot |V(H)|,\;
\alpha(H) \cdot |V(G)|\;\}.
\end{equation}
But this lowerbound can be arbitrarily bad, even for 
cographs~\cite{kn:jha}. For any graph $G$ and any 
natural number $k$ there exists a cograph $H$ 
such that 
\[\alpha(G \times H) \geq k+L(G,H),\] 
where 
$L(G,H)$ is the lowerbound expressed in~\eqref{eqn0}. 
When $G$ and $H$ are vertex transitive then equality holds 
in~\eqref{eqn0}~\cite{kn:zhang}. 

\bigskip 

\begin{definition}
A graph is a cograph if it has no induced 
$P_4$, ie, a path with four vertices. 
\end{definition}
Cographs are characterized by the property that 
every induced subgraph $H$ satisfies one of 
\begin{enumerate}[\rm (a)]
\item $H$ has only one vertex, or  
\item $H$ is disconnected, or 
\item $\Bar{H}$ is disconnected. 
\end{enumerate}
It follows that cographs can be represented 
by a cotree. This is pair $(T,f)$ where $T$ is a 
rooted tree and $f$ is a 1-1 map from the vertices 
of $G$ to the leaves of $T$. Each internal node of $T$, 
including the root,  
is labeled as $\otimes$ or $\oplus$. When the label 
is $\oplus$ then the subgraph $H$, induced by the vertices 
in the leaves, is disconnected. Each child of the 
node represents one component. When the node is labeled 
as $\otimes$ then the complement of the induced 
subgraph $H$ is disconnected. In that case, each component 
of the complement is represented by one child of the node. 

When $G$ is a cograph then a cotree for $G$ can be obtained 
in linear time. 

\bigskip 

Cographs are perfect, see, eg,~\cite[Section~3.3]{kn:kloks2}. 
When $G$ and $H$ are cographs then $G \times H$ is not 
necessarily perfect. For example, when $G$ is the paw, 
ie, $G \simeq K_1 \otimes (K_2 \oplus K_1)$ then $G \times K_3$ 
contains an induced $C_5$~\cite{kn:ravindra}. Ravindra and 
Parthasarathy characterize the pairs $G$ and $H$ for which 
$G \times H$ is perfect~\cite[Theorem~3.2]{kn:ravindra}. 

\section{Independence in categorical products of cographs}
%%%%%%%%%%%%%%%%%%%%%%%%%%%%%%%%%%%%%%%%%%%%%%%%%%%%%%%%%%

It is well-known that $G \times H$ is connected if and only 
if both $G$ and $H$ are connected and at least one of them is not 
bipartite~\cite{kn:weichsel}. When $G$ and $H$ are connected 
and bipartite, then $G \times H$ consists of two components. 
In that case, two vertices $(g_1,h_1)$ and $(g_2,h_2)$ 
belong to the same component if the distances $d_G(g_1,g_2)$ 
and $d_H(h_1,h_2)$ have the same parity.  

\bigskip 

\begin{definition}
The rook's graph $R(m,n)$ is the linegraph of 
the complete bipartite graph $K_{m,n}$. 
\end{definition}
The rook's graph $R(m,n)$ has as its vertices 
the vertices of the grid, 
$(i,j)$, with $1 \leq i \leq m$ and $1 \leq j \leq n$. 
Two vertices are adjacent if they are in the same 
row or column of the grid. 
The rook's graph is perfect, since all linegraphs 
of bipartite graphs are perfect (see, eg,~\cite{kn:kloks2}). 
By the perfect graph theorem, also the complement 
of rook's graph is perfect~\cite{kn:lovasz}. 

\begin{lemma}
\label{lm rook}
Let $m,n \in \mathbb{N}$. Then 
\[K_m \times K_n \simeq \Bar{R},\]
where $\Bar{R}$ is the complement of the 
rook's graph $R=R(m,n)$.
\end{lemma}
\begin{proof}
Two vertices $(i,j)$ and $(i^{\prime},j^{\prime})$ 
are adjacent in $K_m \times K_n$ when $i \neq i^{\prime}$ 
and $j \neq j^{\prime}$. That is, they are adjacent 
when they are not in the same row or column of the 
$m \times n$ grid. Thus, $K_m \times K_n$ is the 
complement of the rook's graph $R(m,n)$. 
\qed\end{proof}

\bigskip 

\begin{lemma}
\label{lm cmp}
Let $G$ and $H$ be complete multipartite. 
Then $G \times H$ is perfect. 
\end{lemma}
\begin{proof}
Ravindra and Parthasarathy prove that $G \times H$ is 
perfect if and only if either 
\begin{enumerate}[\rm (a)]
\item $G$ or $H$ is bipartite, or 
\item Neither $G$ nor $H$ contains an induced odd 
cycle of length at least 5 nor an induced paw. 
\end{enumerate}
Since $G$ and $H$ are perfect, they do not contain 
an odd hole~\cite{kn:chudnovsky}. 
Furthermore, the complement of $G$ and $H$ is a 
union of cliques, and so the complements are $P_3$-free. 
The complement of a paw is $K_1 \oplus P_3$ and so 
it has an induced $P_3$. This proves the claim. 
\qed\end{proof}

\bigskip 

Let $G$ and $H$ be complete multipartite. Let $G$ be the 
join of $m$ independent sets, say with $p_1, \dots, p_m$ vertices, 
and let $H$ be the join of 
$n$ independent sets, say with $q_1,\dots,q_n$ vertices.  
We shortly describe how $G \times H$ is obtained from 
the complement of the rook's graph $R(m,n)$. We call the 
structure a generalized rook's graph.  

\bigskip 

Each vertex $(i,j)$ in $R(m,n)$ is replaced by an 
independent set $I(i,j)$   
of cardinality $p_i \cdot q_j$. Denote the vertices 
of this independent set as 
\[(i_s,j_t) \quad\text{where $1 \leq s \leq p_i$ and 
$1 \leq t \leq q_j$.}\] 
Two vertices $(i_s,j_t)$ and $(i^{\prime}_s,j_t)$ 
are adjacent and these types of row- and column-adjacencies 
are the only adjacencies in this generalized rook's graph. 
The graph $G \times H$ is obtained from the partial complement 
of the generalized rook's graph. 

\bigskip 

\begin{lemma}
Let $G$ and $H$ be complete multipartite graphs. 
Then 
\begin{equation}
\label{eqn1}
\alpha(G \times H) = \kappa(G \times H)=
\max\;\{\;\alpha(G) \cdot |V(H)|,\;\alpha(H) \cdot |V(G)|\;\}.
\end{equation}
\end{lemma}
\begin{proof}
Two vertices $(g_1,h_1)$ and $(g_2,h_2)$ are 
adjacent if $g_1$ and $g_2$ are not in a common 
independent set in $G$ and $h_1$ and $h_2$ are 
not in a common independent set in $H$. 

\medskip 

\noindent
Let $\Omega$ be a maximum independent set of $G$. 
Then 
\[\{\;(g,h)\;|\; g \in \Omega \quad\text{and}\quad h \in V(H)\;\}\] 
is an independent set in $G \times H$. 
We show that all maximal independent sets are of this form 
or of the symmetric form with $G$ and $H$ exchanged. 

\medskip 

\noindent
Consider the complement of the 
rook's graph. Any independent set 
must have all its vertices in one row or in 
one column. 
This shows that every maximal independent set in $G \times H$ 
is a generalized row or column in the rook's graph. 
Since the graphs are perfect, the number of cliques 
in a clique cover of $G \times H$ equals $\alpha(G \times H)$. 
\qed\end{proof}

\bigskip 

\begin{remark}
Notice that complete multipartite graphs are  
not vertex transitive, unless all independent sets 
have the same cardinality. 
\end{remark}

\bigskip 

\begin{lemma}
\label{lm union}
Let $G$ and $H$ be cographs and assume that 
$G$ is disconnected. Say that $G=G_1 \oplus G_2$. 
Then 
\[\alpha(G \times H)=\alpha(G_1 \times H)+\alpha(G_2 \times H).\]
\end{lemma}
\begin{proof}
By definition of the categorical product, 
no vertex of $V(G_1) \times V(H)$ is adjacent to any 
vertex of $V(G_2) \times V(H)$. 
\qed\end{proof}

\begin{lemma}
\label{lm join}
Let $G$ and $H$ be connected cographs. 
Say $G=G_1 \otimes G_2$ and $H=H_1 \otimes H_2$. 
Then 
\[\alpha(G \times H)=\min \;\{\;\alpha(G_1 \times H),\;
\alpha(G_2 \times H),\;\alpha(G \times H_1),\;\alpha(G \times H_2)\;\}.\] 
\end{lemma}
\begin{proof}
Every vertex of $V(G_1) \times V(H_1)$ is adjacent to 
every vertex of $V(G_2) \times V(H_2)$ and, likewise, 
every vertex of $V(G_1) \times V(H_2)$ is adjacent to 
every vertex of $V(G_2) \times V(H_1)$. 
This proves the claim. 
\qed\end{proof}

\bigskip 

\begin{theorem}
There exists a polynomial-time algorithm 
which computes $\alpha(G \times H)$ when $G$ and $H$ 
are cographs. 
\end{theorem}
\begin{proof}
By Lemmas~\ref{lm union} and~\ref{lm join}.
\qed\end{proof}

\section{Splitgraphs}
%%%%%%%%%%%%%%%%%%%%%

F\"oldes and Hammer introduced splitgraphs~\cite{kn:foldes}. 
We refer to~\cite[Chapter~6]{kn:golumbic} and~\cite{kn:merris} for 
some background information on this class of graphs. 

\begin{definition}
A graph $G$ is a splitgraph if there is a partition 
$\{S,C\}$ of its vertices such that $G[C]$ is a clique 
and $G[S]$ is an independent set. 
\end{definition}

\begin{theorem}
Let $G$ and $H$ be splitgraphs. There exists a 
polynomial-time algorithm to compute the independence 
number of $G \times H$. 
\end{theorem}
\begin{proof}
Let $\{S_1,C_1\}$ and $\{S_2,C_2\}$ be the 
partition of $V(G)$ and $V(H)$, respectively, 
into independent sets and cliques. Let 
$c_i=|C_i|$ and $s_i=|S_i|$ for $i \in \{1,2\}$.  
The vertices of $C_1 \times C_2$ form a rook's graph. 

\medskip 

\noindent
We consider three cases. First consider  
the maximum independent 
sets without any vertex of $V(C_1) \times V(C_2)$. 
Notice that the subgraph of $G \times H$ induced 
by the vertices of 
\[V(S_1) \times V(C_2) \cup V(C_1) \times V(S_2) \cup V(S_1) \times V(S_2)\] 
is bipartite. A maximum independent set in a bipartite graph can be 
computed in polynomial time. 

\medskip 

\noindent
Consider maximum independent sets that contain exactly one 
vertex $(c_1,c_2)$ of $V(C_1) \times V(C_2)$. The maximum 
independent set of this type can be computed as follows. 
Consider the bipartite graph of 
the previous case and remove the neighbors of $(c_1,c_2)$ 
from this graph. 
The remaining graph is bipartite. Maximizing over all 
pairs $(c_1,c_2)$ gives the maximum independent set of this 
type. 

\medskip 

\noindent
Consider maximum independent sets that contain at least 
two vertices of the rook's graph $V(C_1) \times V(C_2)$. 
Then the two vertices must be in one row or in one 
column of the grid, since otherwise they are adjacent. 
Let the vertices of the independent set be contained 
in row $c_1 \in V(C_1)$. Then the vertices of 
$V(S_1) \times V(C_2)$ of the independent set 
are contained in 
\[W=\{\;(s_1,c_2)\;|\; s_1 \notin N_G(c_1) \quad\text{and}\quad 
c_2 \in C_2\;\}.\] 
Consider the bipartite graph with one color class 
defined as the following set of vertices 
\[\{\;(c_i,s_2)\;|\; c_i \in C_1 \;\text{and}\; s_2 \in S_2\;\}
\cup \{\;(s_1,s_2)\;|\; 
s_1 \in V(S_1) \;\text{and}\; s_2 \in V(S_2)\;\},\]   
and the other color class defined as  
\[W \cup \{\;(c_1,c_2)\;|\; c_2 \in C_2\;\}.\]  
Since this graph is bipartite, the maximum independent set 
of this type can be computed in polynomial time by maximizing 
over the rows $c_1 \in C_1$ and columns $c_2 \in C_2$. 

\medskip 

\noindent
This proves the theorem. 
\qed\end{proof}

\section{Tensor capacity}
%%%%%%%%%%%%%%%%%%%%%%%%%

In this section we consider the powers of a graph 
under the categorical product. 

\bigskip 

\begin{definition}
The independence ratio of a graph $G$ is 
defined as 
\begin{equation}
\label{eqn2}
i(G) = \frac{\alpha(G)}{|V(G)|}.
\end{equation}
\end{definition}
For background information on the related Hall-ratio we 
refer to~\cite{kn:simonyi,kn:toth3,kn:toth5,kn:toth4}. 

\bigskip 

By~\eqref{eqn0} for any two graphs $G$ and $H$ we 
have 
\begin{equation}
\label{eqn3}
i(G \times H) \geq \max\;\{\;i(G),\;i(H)\;\}.
\end{equation}
It follows that $i(G^k)$ is non-decreasing. Also, it is bounded 
from above by 1 and so the limit when $k \rightarrow \infty$ exists. 
This limit was introduced in~\cite{kn:brown} as the 
`ultimate categorical independence ratio.' See 
also~\cite{kn:alon,kn:hahn2,kn:hell2,kn:lubetzky}. 
For simplicity we call it the 
tensor capacity of a graph. 
Alon and Lubetzky, and also T\'oth claim that computing the 
tensor capacity is NP-complete, but neither 
provides a proof~\cite{kn:alon,kn:lubetzky,kn:toth2,kn:toth4}. 

\begin{definition}
Let $G$ be a graph. The tensor capacity of $G$ is 
\begin{equation}
\label{eqn4}
\Theta^T(G) = \lim_{k \rightarrow \infty} i(G^k).
\end{equation}
\end{definition}

\bigskip 

Brown et al.~\cite[Theorem 3.3]{kn:brown} 
obtain the following lowerbound for the 
tensor capacity. 
\begin{equation}
\label{eqn5}
\Theta^T(G) \geq a(G) 
\quad\text{where}\quad a(G)=\max_{\text{$I$ an independent set}}\; 
\frac{|I|}{|I| + |N(I)|}. 
\end{equation}
It is related to the binding 
number $b(G)$ of the graph $G$. Actually, the binding number 
is less than 1 if and only if $a(G) > \frac{1}{2}$. 
In that case, the binding number is realized by an independent set 
and it is equal to $b(G)=\frac{1-a(G)}{a(G)}$~\cite{kn:kloks,kn:toth2}. 
The binding number is computable in 
polynomial time~\cite{kn:cunningham,kn:kloks,kn:woodall}. 
See also Corollary~\ref{cor pol a^T} below.    

\bigskip 

The following proposition was proved in~\cite{kn:brown}. 

\begin{proposition}
If $i(G) > \frac{1}{2}$ then $\Theta^T(G)=1$. 
\end{proposition}
Therefore, a better lowerbound for $\Theta^T(G)$ is 
provided by 
\begin{equation}
\label{eqn6}
\Theta^T(G) \geq a^{\ast}(G)=\begin{cases}
a(G) & \quad\text{if $a(G) \leq \frac{1}{2}$}\\
1 & \quad\text{if $a(G) > \frac{1}{2}$.}
\end{cases}
\end{equation}

\bigskip 

\begin{definition}
Let $G=(V,E)$ be a graph. A fractional matching 
is a function $f: E \rightarrow \mathbb{R}^{+}$, which assigns 
a non-negative real number to each edge, such that 
for every vertex $x$ 
\[\sum_{e \ni x} \; f(e) \leq 1.\]
A fractional matching $f$ is perfect if it achieves the 
maximum 
\[f(E)=\sum_{e \in E}\; f(e) = \frac{|V|}{2}.\] 
\end{definition}
 
Alon and Lubetzky proved the following theorem in~\cite{kn:alon} 
(see also~\cite{kn:kloks}). 

\begin{theorem}
For every graph $G$ 
\begin{equation}
\label{eqn7}
\Theta^T(G)=1 \;\Leftrightarrow\; a^{\ast}(G)=1 \;
\Leftrightarrow\; \text{$G$ has no fractional perfect matching.}
\end{equation}
\end{theorem}

\begin{corollary}
\label{cor pol a^T}
There exists a polynomial-time algorithm to decide 
whether 
\[\Theta^T(G) =1 \quad\text{or}\quad \Theta^T(G) \leq \frac{1}{2}.\]  
\end{corollary}

\bigskip 
 
The following theorem was raised as a 
question by Alon and Lubetzky  
in~\cite{kn:alon,kn:lubetzky}. The theorem was proved by 
\'Agnes T\'oth~\cite{kn:toth2}. 

\begin{theorem}
\label{thm toth}
For every graph $G$ 
\[ \Theta^T(G)=a^{\ast}(G).\]
Equivalently, every graph $G$ satisfies  
\begin{equation}
\label{eqn8}
a^{\ast}(G^2)=a^{\ast}(G). 
\end{equation}
\end{theorem}
T\'oth proves that 
\begin{equation}
\label{eqn11}
\text{if $a(G) \leq \frac{1}{2}$ or $a(H) \leq \frac{1}{2}$ then}
\quad 
a(G \times H) \leq \max\;\{\;a(G),\;a(H)\;\}.
\end{equation}
Actually, T\'oth shows that, if $I$ is an independent 
set in $G \times H$ then   
\[|N_{G \times H}(I)| \geq |I| \cdot \min\;\{\;b(G),\;b(H)\;\}.\]
{F}rom this, Theorem~\ref{thm toth} easily follows. 
As a corollary 
(see~\cite{kn:alon,kn:lubetzky,kn:toth2})   
one obtains that, 
for any two graphs $G$ and $H$
\[i(G \times H) \leq \max\;\{\;a^{\ast}(G),\;a^{\ast}(H)\;\}.\] 

\bigskip 

T\'oth also proves the 
following theorem in~\cite{kn:toth2}. 
This was conjectured by 
Brown et al.~\cite{kn:brown}. 

\begin{theorem} 
\label{thm tensor cap union}
For any two graphs $G$ and $H$, 
\begin{equation}
\label{eqn9}
\Theta^T(G \oplus H)=\max\;\{\;\Theta^T(G),\;\Theta^T(H)\;\}.
\end{equation}
\end{theorem}

Notice that the analogue of this statement, 
with $a^{\ast}$ 
instead of $\Theta^T$, is straightforward. The theorem  
follows from~\eqref{eqn11} via the following lemma.  
This lemma was proved by 
Alon and Lubetzky in~\cite{kn:alon}. 

\begin{lemma}
\label{lm tensor cap join}
 For any two graphs $G$ and $H$, 
\begin{equation}
\label{eqn10}
\Theta^T(G \oplus H) = \Theta^T(G \times H).
\end{equation}
\end{lemma}

\bigskip 

T\'oth proves the following theorem in~\cite[Corollary 3]{kn:toth}.
This is proved as a corollary of a 
theorem which says that, 
if for all $x \in V$, $d(x) \geq n-\alpha(G)$, and if 
$i(G) \leq \frac{1}{2}$, then 
for all $k \in \mathbb{N}$,  
\[i(G^k)=i(G).\] 
  
\begin{theorem}
Let $G$ be a complete multipartite graph. 
Let $\alpha$ be the size of the largest partite class of $G$. 
Then 
\[\Theta^T(G)= 
\begin{cases} 
i(G)=\frac{\alpha}{n} & \quad\text{if $\alpha \leq \frac{n}{2}$}\\
1 & \quad\text{otherwise.}
\end{cases}\] 
\end{theorem}

\bigskip 

For cographs we obtain the following theorem. 

\begin{theorem}
There exists a polynomial-time algorithm to compute 
the tensor capacity for cographs. 
\end{theorem}
\begin{proof}
By Theorem~\ref{thm toth} it is sufficient to 
compute $a(G)$, as defined in~\eqref{eqn5}. 

\medskip 

\noindent 
Consider a cotree for $G$. For each node the algorithm 
computes a table. The table contains numbers  
$\ell(k)$, for $k \in \mathbb{N}$, where 
\[\ell(k)= min\;\{\;|N(I)|\;|\; \text{$I$ is an independent set with 
$|I|=k$}\;\}.\] 
Notice that $a(G)$ can be obtained from 
the table at the root node via 
\[a(G)=\max_k \; \frac{k}{k+\ell(k)}.\] 

\medskip 

\noindent
Assume $G$ is the union of two cographs $G_1 \oplus G_2$. 
An independent set $I$ is the union of two 
independent sets $I_1$ in $G_1$ and $I_2$ in $G_2$. 
Let the table entries for $G_1$ and $G_2$ be 
denoted by the functions $\ell_1$ and $\ell_2$.  
Then 
\[\ell(k)=\min\;\{\;\ell_1(k_1) + \ell_2(k_2)\;|\; k_1+k_2=k\;\}.\] 

\medskip 

\noindent
Assume that $G$ is the join of two cographs, say 
$G = G_1 \otimes G_2$. 
An independent set in $G$ can have vertices in at most 
one of $G_1$ and $G_2$. Therefore, 
\[\ell(k) = \min\;\{\;\ell_1(k)+|V(G_2)|,\; 
\ell_2(k)+|V(G_1)|\;\}.\] 

\medskip 

\noindent
This proves the theorem. 
\qed\end{proof}

\begin{remark}
The tensor capacity is computable in 
polynomial time for many other classes of graphs 
via similar methods~\cite{kn:kratsch}. 
\end{remark}

\section{Concluding remarks}
%%%%%%%%%%%%%%%%%%%%%%%%%%%%

It would be interesting to know whether the 
tensor capacity for splitgraphs is computable in polynomial time. 
Also, is the independence number for the product of 
three splitgraphs, $G_1 \times G_2 \times G_3$ NP-complete?

\end{document}